\documentclass[12pt,journal,onecolumn]{IEEEtran}

\usepackage{epsfig}
\usepackage{epstopdf}
\usepackage{graphicx}
\usepackage{amsmath}
\usepackage{amsthm}
\usepackage{amsfonts}
\usepackage{cite}
\usepackage{amssymb}
\usepackage{amsbsy}
\usepackage{float}
\usepackage{balance}
\usepackage[stable]{footmisc}
\newtheorem{theorem}{\bf Theorem}

\title{Performance Analysis of Hybrid ARQ with Incremental Redundancy over Amplify-and-Forward Dual-Hop Relay Channels\\[0ex]}
\author{{Ali~Chelli,~\IEEEmembership{Member,~IEEE,}~Amir Hadjtaieb,~and~Mohamed-Slim~Alouini,~\IEEEmembership{Fellow,~IEEE}}}

\begin{document}
\maketitle
\linespread{2}\selectfont
\begin{abstract}
\let\thefootnote\relax\footnotetext{ A. Chelli, A. Hadjtaieb, and M.-S. Alouini are with the Computer, Electrical, and Mathematical Science and Engineering (CEMSE) Division, King Abdullah University of Science and Technology (KAUST) Thuwal, Makkah Province, Saudi Arabia (e-mails:\{ali.chelli,amir.hadjtaieb,slim.alouini\}@kaust.edu.sa).
}
In this paper, we consider a three node relay network comprising a source, a relay, and a destination. The source transmits the message to the destination using hybrid automatic repeat request (HARQ) with incremental redundancy (IR). The relay overhears the transmitted message, amplifies it using a variable gain amplifier, and then forwards the message to the destination. This latter combines both the source and the relay message and tries to decode the information. In case of decoding failure, the destination sends a negative acknowledgement. A new replica of the message containing new parity bits is then transmitted in the subsequent HARQ round. This process continues until successful decoding occurs at the destination or a maximum number $M$ of rounds is reached. We study the performance of HARQ-IR over the considered relay channel from an information theoretic perspective. We derive exact expressions and bounds for the information outage probability, the average number of transmissions, and the average transmission rate. Moreover, we evaluate the delay experienced by Poisson arriving packets over the considered relay network. We also provide analytical expressions for the expected waiting time, the sojourn time, and the energy efficiency. The derived exact expressions are validated by Monte Carlo simulations.
\end{abstract}
\begin{IEEEkeywords}
Amplify-and-forward relaying, hybrid automatic repeat request (HARQ), incremental redundancy, energy efficiency, delay analysis, information outage capacity, average transmission rate.
\end{IEEEkeywords}
\section{Introduction}
During the last decade, relay networks have attracted a lot of interest due to their numerous benefits. The relaying technique allows extending the coverage zone of wireless networks and offers a higher reliability for communication systems \cite{Laneman2004a}. Relay networks have been vastly investigated in the literature \cite{Hasna2003} and they have been as well adopted in wireless communication standards \cite{Oyman2007}. The performance of relay networks can be improved further by the use of the hybrid automatic repeat request (HARQ) technique. HARQ combines both forward error correction and retransmission mechanisms which yields a larger spectral efficiency and ensure a good quality of service even in absence of channel state information at the transmitter.

There are two major HARQ techniques, namely, HARQ with incremental redundancy (IR) and HARQ with Chase combining (CC). HARQ-IR has a larger spectral efficiency relatively to HARQ-CC. However, from an implementation perspective HARQ-IR is more complex and requires more hardware resources \cite{Chelli2013}. In the literature, the performance of HARQ over relay networks has been investigated from an information theoretic perspective in few papers \cite{Narasimhan2008,Maham2012,Chelli2013}. More specifically, the author of \cite{Narasimhan2008} analyzed the outage capacity and the throughput of HARQ-IR over a three node relay network. In \cite{Narasimhan2008}, the channel is assumed to be constant during all HARQ rounds. This assumption simplifies significantly the analysis. However, this assumption is only valid for low mobility scenarios. In \cite{Maham2012}, a similar assumption is considered. The authors of \cite{Maham2012} consider a scenario with multiple relays. At each HARQ round the best relay is selected in a distributed manner. The delay-limited throughput of this communication system is explored in \cite{Maham2012}. In \cite{Chelli2013}, the outage capacity of HARQ-IR and HARQ-CC are studied over a three node relay network. In that paper, the channel is assumed to be variable and changes independently from one HARQ round to another. It is worth mentioning that the considered relay networks in \cite{Narasimhan2008,Maham2012,Chelli2013} operate in a decode-and-forward mode.

In this paper, we study the performance of HARQ-IR over a three node relay network. As opposed to \cite{Narasimhan2008,Maham2012,Chelli2013}, we consider in our investigation amplify-and-forward relaying (with variable gain at the relay). It is well known that the analysis of amplify-and-forward relay networks is substantially different relatively to decode-and-forward relay channels. We assume that the channel varies independently for different HARQ rounds which makes our analysis more general. We derive analytical expressions for the information outage probability of HARQ-IR over an amplify-and-forward relay network. In addition, we provide expressions for the average number of transmissions and the average transmission rate. In order to limit the delay, a maximum number $M$ of HARQ rounds is assumed. Moreover, we analyze the delay performance of HARQ-IR. We provide analytical expressions for the average waiting time for a data packet (the average time elapsed between the first packet transmission and the packet arrival) and the average sojourn time in the buffer (the average waiting time in the buffer for a data packet before the start of its transmission). We study as well the energy efficiency of HARQ-IR over the amplify-and-forward relay channel.

The remainder of this paper is organized as follows. Section~II describes the system model. An exact expression and an upper bound for the outage probability are presented in Section~III. The average number of transmissions and the average transmission rate are investigated in Section~IV and V, respectively. Section~VI is devoted to the study of the energy efficiency of HARQ-IR over amplify-and-forward relay channels. The delay analysis of this scheme is carried out in Section~VII. Finally, we draw the conclusion in Section~VIII.
\section{System Model and Preliminaries}
We consider a three node relay channel comprising a source, a relay, and a destination. The wireless link between any two nodes is modeled as a Rayleigh fading channel. Therefore, the squared envelop of the channel gains associated with the source-destination $\gamma_{SD}$, the source-relay $\gamma_{SR}$, and the relay-destination $\gamma_{RD}$ are independent exponentially distributed stochastic processes with mean values $\bar{\gamma}_{SD}$, $\bar{\gamma}_{SR}$, and $\bar{\gamma}_{RD}$, respectively. We assume that channel realizations between any two nodes are constant during one HARQ round, but change independently in subsequent HARQ rounds. To limit the delay and avoid network congestion, a maximum number $M$ of HARQ rounds is assumed. In each HARQ round, the source transmits the message to the destination and the relay overhears it. Subsequently and in a second step, the relay amplifies and forwards the message to the destination which combines both signals using maximal ratio combining and then decodes the information. An illustration of the protocol is provided in Fig.~\ref{fig:systemModel}.
\begin{figure}[h]
\begin{center}
\scalebox{0.5}{\includegraphics{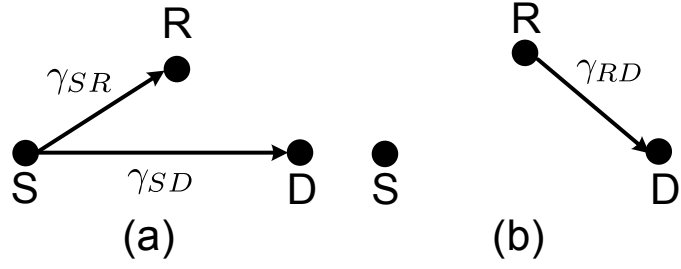}}
\caption{\label{fig:systemModel}Illustration of a single HARQ round: (a) Phase 1: the source transmits the message to the destination and the relay overhears it. (b) Phase 2: the relay amplify and forward the message to the destination. This latter combines the source message with the relayed message and tries to decode the information.\vspace*{-0.5cm}}
\end{center}
\end{figure}

Let us assume that $b$ information bits are transmitted in each data packet. When using HARQ-IR, these $b$ bits are encoded at the source into $M \times L$ symbols, where $L$ is the number of transmitted symbol by the source in every HARQ round (the source and the relay transmit together $2L$ symbols in each HARQ round). The number $m$ of HARQ rounds required to transmit successfully a data packet varies randomly according to the channel conditions. If the channel conditions are good, one HARQ round could be sufficient. Otherwise, $M$ HARQ round might be required to transmit a data packet. The transmission rate used by the source during the first HARQ round can be expressed as $R_1=\frac{b}{L}$ bits per channel use. Taking into account the relay transmission, we would have a rate of $\frac{R_1}{2}$ for the first HARQ round. If $m$ HARQ rounds are used to transmit one data packet, then the transmission rate is equal to $R_m=\frac{R_1}{2m}$. Each data packet containing $b$ information bits is encoded into $m \times 2L$ symbols. As a result, the robustness of the codeword toward fading depends on $m$. Clearly, the most robust codeword corresponds to the case when $M$ HARQ round are used to transmit a given data packet. However, this last case corresponds also to the lowest spectral efficiency. At the receiver side, the symbols received during all HARQ rounds are combined before decoding.

The capacity in bits/symbols of HARQ-IR after $m$ rounds can be expressed as
\begin{equation}\label{Capacity}
C_m=\frac{1}{2m} \sum_{l=1}^m \log_2(1+\gamma_l),
\end{equation}
where $\gamma_l$ represents the received signal at the destination during the $l$th transmission round. The prelog factor $1/2$ is due to the use of two time slots in each HARQ round (one time slot for the source transmission and one time slot for the relay transmission). In (\ref{Capacity}), $\gamma_l$ is the sum of the source-destination signal and the source-relay-destination signal, i.e., $\gamma_l=\gamma_{SRD}+\gamma_{SD}$. By using the same approximation adopted in \cite{Anghel2004,Bletsas2006,Ikki2011}, the equivalent signal to noise ratio (SNR) $\gamma_{SRD}$ of the link source-relay-destination can be approximated by its upper bound $\gamma_{up}$ as
\begin{equation}
\gamma_{SRD}= \frac{\gamma_{SR} \gamma_{RD}}{1+\gamma_{SR}+\gamma_{RD}} \leq \gamma_{up}=\mathrm{min}(\gamma_{SR},\gamma_{RD}),
\end{equation}
where $\min(x,y)$ is the minimum value of $x$ and $y$. This approximation is shown to be very tight and analytically more tractable than the exact value $\gamma_{SRD}$. Since $\gamma_{SR}$ and $\gamma_{RD}$ are independent exponentially distributed, the probability density function (PDF) of min($\gamma_{SR}$,$\gamma_{RD}$) is also exponential with a mean $\bar{\gamma}_{up}=\frac{\bar{\gamma}_{SR}\bar{\gamma}_{RD}}{\bar{\gamma}_{SR}+\bar{\gamma}_{RD}}$. Thus, $\gamma_l$ is the sum of two independent exponential distributed random variables with means $\bar{\gamma}_{SD}$ and $\bar{\gamma}_{up}$. It follows that the PDF of $\gamma_l$ can be expressed as
\begin{equation}\label{PDFgammal}
p_{\gamma_l}(\gamma)= C \left(  \mathrm{e}^{-\frac{\gamma}{\bar{\gamma}_{up}}}-\mathrm{e}^{-\frac{\gamma}{\bar{\gamma}_{SD}}} \right),~~~\gamma \geq 0,
\end{equation}
where
\begin{equation}
C=\frac{\bar{\gamma}_{SR}+\bar{\gamma}_{RD}}{\bar{\gamma}_{SR}\bar{\gamma}_{RD}-\bar{\gamma}_{SD}(\bar{\gamma}_{SR}+ \bar{\gamma}_{RD})}=\frac{1}{\bar{\gamma}_{up}-\bar{\gamma}_{SD}}.
\end{equation}
Since channel gains are independent through the different transmission rounds, the random variables $\gamma_l(l=1,\dots,M)$ are independent and identically distributed (i.i.d.).

\section{Outage Probability Analysis}

The system is in outage if the destination fails to decode the packet after $M$ HARQ rounds, or equivalently, if the capacity $C_M$ is less than the rate $R_M$. Thus, an outage happens if the accumulated mutual information after $M$ rounds is less than the rate $R_1$. As a result, the outage probability of HARQ-IR after $M$ rounds can be written as
\begin{align}
P_{out}&=\mathrm{Pr}\Big[\sum_{l=1}^M \log_2(1+\gamma_l) \leq R_1\Big] \nonumber \\
&=\mathrm{Pr}\Big[ \log_2 \Big( \prod_{l=1}^M (1+\gamma_l) \Big) \leq R_1\Big].
\end{align}
We will provide an exact expression for outage probability in Section~A and a tight upper bound in Section~B.

\subsection{Exact Expression}
In this section, we derive an exact expression for outage probability of HARQ-IR over amplify-and-forward dual-hop relay networks. The outage probability can be written as
\begin{align}
&P_{out}=\mathrm{Pr}\Big[ \log_2 \Big( \prod_{l=1}^M (1+\gamma_l) \Big) \leq R_1\Big] =\mathrm{Pr}\Big[ \prod_{l=1}^M  Z_l \leq 2^{R_1} \Big],
\end{align}
where $Z_l=1+\gamma_l$.
\begin{theorem}\label{Theo1}
Let $\xi_M=\prod_{l=1}^M Z_l$. The expression of the PDF of $\xi_M$ can be determined as
\begin{align}
&p_{\xi_M}(x) = C^M \sum_{k=0}^M \binom{M}{k} (-1)^{M-k} \mathrm{e}^{-\frac{k}{\bar{\gamma}_{up}}}~ \mathrm{e}^{-\frac{M-k}{\bar{\gamma}_{SD}}} \times \nonumber\\
&H_{0,M}^{M,0} \Bigg[~ \frac{x}{\bar{\gamma}_{up}^{k}~ \bar{\gamma}_{SD}^{(M-k)}} \Bigg\vert \begin{array} {c} - \\ \Phi_1,\ldots,\Phi_k,\Psi_1,\ldots,\Psi_{M-k} \\ \end{array} \Bigg],
\end{align}
where the terms $\Phi_{\ell}$ and $\Psi_p$ are defined as
\begin{eqnarray}
  \Phi_{\ell} &\triangleq& (0,1,\frac{1}{\bar{\gamma}_{up}})\,\,\, \mathrm{for} \,\,\,\, \ell=1,\ldots,k\\
  \Psi_p&\triangleq& (0,1,\frac{1}{\bar{\gamma}_{SD}})\,\,\, \mathrm{for} \,\,\,\, p=1,\ldots,M-k.
\end{eqnarray}
The function $H_{p,q}^{m,n}(\cdot)$ stands for the generalized upper incomplete Fox's H function defined as \cite{ferkan2009}
\begin{align}
&H_{p,q}^{m,n} \Bigg[~ z \Bigg\vert \begin{array} {c} (a_1,\alpha_1,A_1),\dots,(a_p,\alpha_p,A_p) \\ (b_1,\beta_1,B_1),\dots,(b_q,\beta_q,B_q) \\ \end{array} \Bigg] \nonumber\\
&= \frac{1}{j2\pi} \oint_\varsigma \frac{\prod_{i=1}^m \Gamma(b_i+\beta_i s,B_i) }{\prod_{i=n+1}^p \Gamma(a_i+\alpha_i s, A_i)} \times   \nonumber\\
&\hspace{2.5cm}\frac{ \prod_{i=1}^n \Gamma(1-a_i-\alpha_i s,A_i)}{ \prod_{i=m+1}^q \Gamma(1-b_i-\beta_i s,B_i)} z^{-s} ds,
\end{align}
where $\Gamma(\cdot,\cdot)$ is the upper incomplete gamma function \cite[Eq.(6.5.3)]{Abramowitz1964}.
\end{theorem}
\begin{proof}
See Appendix~\ref{A.A}.
\end{proof}

\begin{theorem}\label{Theo2}
The expression of the cumulative distribution function (CDF) of $\xi_M$ can be derived as
\begin{align}
\nonumber&F_{\xi_M}(x) 
=C^M \sum_{k=0}^M \binom{M}{k} \bar{\gamma}_{up}^k ~(-\bar{\gamma}_{SD})^{M-k} ~ \mathrm{e}^{-\frac{k}{\bar{\gamma}_{up}}} \mathrm{e}^{-\frac{M-k}{\bar{\gamma}_{SD}}} \times \nonumber \\
& H_{1,M+1}^{M,1} \Bigg[\frac{x}{\bar{\gamma}_{up}^{k}~ \bar{\gamma}_{SD}^{(M-k)}} \Bigg\vert   \begin{array} {c} (1,1,0) \\ \Phi'_1,\ldots,\Phi'_k,\Psi'_1,\ldots,\Psi'_{M-k},(0,1,0) \\ \end{array} \Bigg],
\end{align}
where the terms $\Phi_{\ell}$ and $\Psi_p$ are defined as
\begin{eqnarray}
  \Phi'_{\ell} &\triangleq& (1,1,\frac{1}{\bar{\gamma}_{up}})\,\,\, \mathrm{for} \,\,\,\, \ell=1,\ldots,k\\
  \Psi'_p&\triangleq& (1,1,\frac{1}{\bar{\gamma}_{SD}})\,\,\, \mathrm{for} \,\,\,\, p=1,\ldots,M-k.
\end{eqnarray}
\end{theorem}
\begin{proof}
See Appendix~\ref{A.B}.
\end{proof}
Finally, the outage probability can be written as
\begin{equation}
\label{Pout}
P_{out} = F_{\xi_M}(2^{R_1}).
\end{equation}
\subsection{Upper Bound Expression}
In this section, we derive an upper bound expression for outage probability. This upper bound has a simpler expression than the exact outage probability. Actually, the upper bound and the exact expression are expressed in terms of Meijer's G-function and the the generalized upper incomplete Fox's H function, respectively. Note that the Meijer's G-function is readily available in Mathematica, while the generalized upper incomplete Fox's H function needs to be implemented utilizing the code provided in \cite{ferkan2009}.
Using the following bound derived from the Minkowski inequality
\begin{equation}
\big((\prod_{i=1}^M x_i)^{\frac{1}{M}} +1\big)^M \leq \prod_{i=1}^M (x_i+1)~~~~\mathrm{for}~~~ x_i \geq 0,
\end{equation}
we obtain an upper bound for outage probability which can be expressed as
\begin{align}
P_{out} \leq \bar{P}_{out}&= \mathrm{Pr}\Big[M \log_2 \big(1+(\prod_{l=1}^M \gamma_l)^{\frac{1}{M}}\big) \leq R_1 \Big] \nonumber \\
&=\mathrm{Pr}\big[\prod_{l=1}^M \gamma_l \leq (2^{\frac{R_1}{M}} -1)^M \big].
\end{align}
\begin{theorem}\label{Theo3}
Let $\zeta_M=\prod_{l=1}^M \gamma_l$. The expression of the PDF of $\zeta_M$ can be evaluated as
\begin{align}
p_{\zeta_M}(x) =& C^M \sum_{k=0}^M \binom{M}{k} (-1)^{M-k}
G_{0,M}^{M,0} \Bigg[~ \frac{x}{\bar{\gamma}_{up}^{k}~ \bar{\gamma}_{SD}^{(M-k)}} \Bigg\vert \begin{array} {c} - \\ 0,...,0 \\ \end{array} \Bigg],
\end{align}
where $G_{p,q}^{m,n}(\cdot)$ is the Meijer's G-function defined as \cite{Luke1969}
\begin{align}
\label{MeijerG}
&G_{p,q}^{m,n} \Bigg[~ z \Bigg\vert \begin{array} {c} a_1,\dots,a_p \\ b_1,\dots,b_q \\ \end{array} \Bigg]  \nonumber\\
&= \frac{1}{j2\pi} \oint_\varsigma \frac{\prod_{i=1}^m \Gamma(b_i+s) \prod_{i=1}^n \Gamma(1-a_i-s)}{\prod_{i=n+1}^p \Gamma(a_i+s) \prod_{i=m+1}^q \Gamma(1-b_i-s)} z^{-s} ds.
\end{align}
\end{theorem}
\begin{proof}
See Appendix~\ref{B.A}.
\end{proof}
\begin{theorem}\label{Theo4}
The expression of the CDF of $\zeta_M$ can be determined as
\begin{align}
F_{\zeta_M}(x) &=C^M \sum_{k=0}^M \binom{M}{k} (-1)^{M-k}~ \bar{\gamma}_{up}^{k}~ \bar{\gamma}_{SD}^{(M-k)} \times \nonumber\\
&G_{1,M+1}^{M,1} \Bigg[\frac{x}{\bar{\gamma}_{up}^{k}~ \bar{\gamma}_{SD}^{(M-k)}}  \Bigg\vert \begin{array} {c} 1 \\ 1,\dots,1,0 \\ \end{array} \Bigg].
\end{align}
\end{theorem}
\begin{proof}
See Appendix~\ref{B.B}.
\end{proof}
Finally, an upper bound for outage probability can be expressed as
\begin{equation}
\label{barPout}
\bar{P}_{out} = F_{\zeta_M}((2^{\frac{R_1}{M}} -1)^M).
\end{equation}
\subsection{Numerical Results}
In this section, the analytical expressions for the outage probability for HARQ-IR over amplify-and-forward dual-hop relay networks are evaluated numerically and validated by Monte Carlo simulations. Both source and relay are assumed to transmit with the same power $P$. The squared envelop of the channel gain associated with the source-destination, the source-relay, and the relay-destination links are exponentially distributed random variables with mean values $\bar{\gamma}_{SD}=0.1$ , $\bar{\gamma}_{SR}=0.3$, and $\bar{\gamma}_{RD}=0.5$, respectively.

To show the impact of cooperation on HARQ-IR performance, we plot the outage probability of HARQ-IR with and without relaying for different values of $M$ in Fig.~\ref{fig:OutagePrWithoutRel}. This figure shows that the gain realized thanks to relaying is large. Indeed, relaying achieves a 6 dB gain an outage probability of 2 $10^{-2}$ for $M=4$ and $R_1=2$ bps/Hz. Additionally, we notice that this gain increases as the number of HARQ rounds $M$ increases. It can also be seen that the outage probability decreases as $M$ increases.
\begin{figure}[h]  \centerline{\psfig{figure=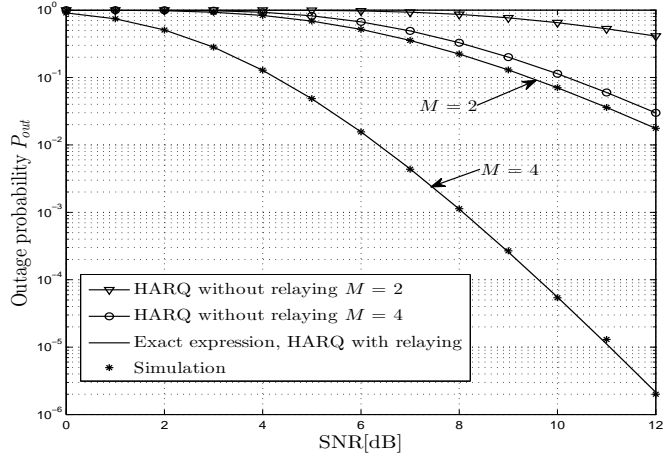,height=2.5 in,width=3.60in} }
    \caption{System outage probability versus SNR with and without relaying for $R_1=~2$~bps/Hz.\vspace*{-0.4cm}}
    \label{fig:OutagePrWithoutRel}
\end{figure}

Fig.~\ref{fig:OutagePrDifferentRate} depicts the outage probability $P_{out}$ versus SNR for different rate values $R_1$. We set the maximum number of HARQ rounds $M=2$. It can be seen that the outage probability increases as rate $R_1$ increases. We observe also that the upper bound expression of outage probability is tight, especially at low SNR, and becomes tighter as the rate increases.
\begin{figure}[h]  \centerline{\psfig{figure=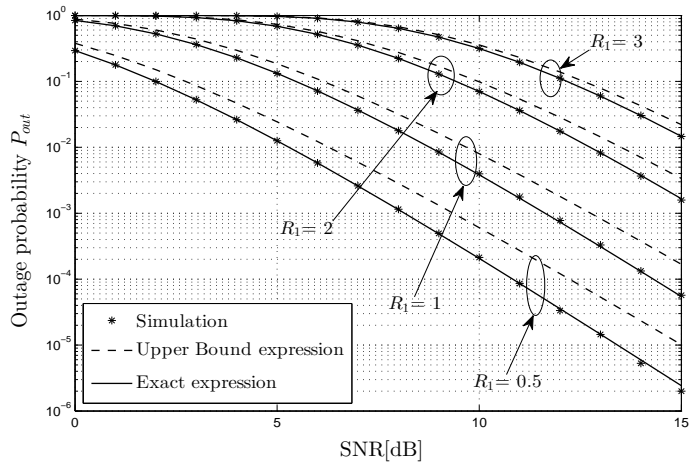,height=2.5 in,width=3.60in} }
    \caption{System outage probability versus SNR for different rate values ($R_1$=~0.5, 1, 2, 3~bps/Hz) and $M=2$.\vspace*{-0.4cm}}
    \label{fig:OutagePrDifferentRate}
\end{figure}
\section{Average Number of Transmissions}

\subsection{Analytical Expression}
Let $P_{out}(m)$ be the probability of a decoding failure at the destination after $m$ HARQ rounds. The average number of transmissions per data packet $\bar{N}$ can be expressed as \cite{ChelliPatzold2011,ChelliBarry2011}
\begin{equation}
\label{Nbar}
\bar{N}= 1+ \sum_{m=1}^{M-1} P_{out}(m).
\end{equation}
Using (\ref{Pout}), the expression of the average number of transmissions $\bar{N}$ can be written as
\begin{equation}
\bar{N}= 1+ \sum_{m=1}^{M-1} F_{\xi_m}((2^{R_1}).
\end{equation}
Utilizing (\ref{barPout}), an upper bound expression for the average number of transmissions $\bar{N}$ can be derived as
\begin{equation}
\bar{N}^{up}= 1+ \sum_{m=1}^{M-1} F_{\zeta_m}((2^{\frac{R_1}{m}} -1)^m).
\end{equation}

\subsection{Numerical Results}

As a numerical example, Fig.~\ref{fig:NbarDifferentRate} illustrates the average number of transmissions $\bar{N}$ versus SNR for different rate values and a maximum number of HARQ rounds $M=3$. This figure shows that the average number of transmissions $\bar{N}$ decreases when the SNR increases. Additionally, for high SNR, just one HARQ round is enough to decode the message correctly at the receiver side. However, for low SNR, the maximum number of HARQ rounds $M=3$ is used. It can be seen also that given a fixed SNR, $\bar{N}$ increases when rate increases. Hence, when we increase the rate, the message requires more HARQ rounds to be decoded successfully at the destination. Finally, we notice that the upper bound expression for outage probability gives a tight upper bound for the average number of transmissions $\bar{N}$ also.
\begin{figure}[h]  \centerline{\psfig{figure=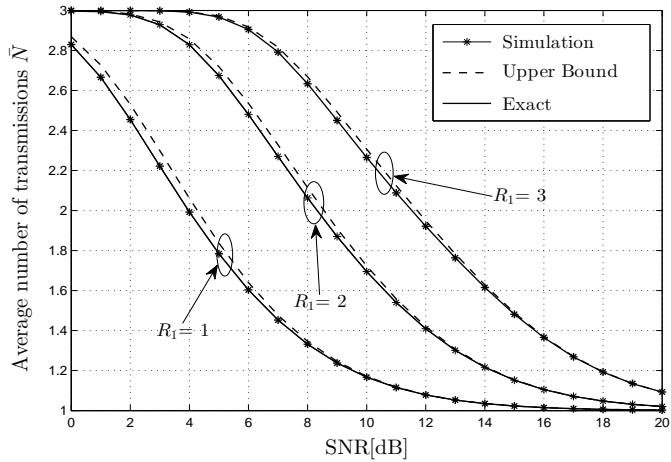,height=2.5 in,width=3.6 in} }
    \caption{Average number of transmissions $\bar{N}$ versus SNR for different rate values ($R_1$=~1, 2, 3~bps/Hz) and $M=3$.\vspace*{-0.4cm}}
    \label{fig:NbarDifferentRate}
\end{figure}
\section{Average Transmission Rate}
\subsection{Analytical Expression}
Let $N$ be the number of packets to be transmitted and $b$ the number of information bits in each packet. In each HARQ round, $L$ symbols are transmitted. We denote by $Q_n$ the number of HARQ rounds required for an error-free transmission of the $n$th data packet. The average transmission rate for $N$ data packets can be written as
\begin{equation}
\bar{R} = \frac{N b}{2L \sum_{n=1}^N Q_n} = \frac{R_1}{\frac{2}{N} \sum_{n=1}^N Q_n} = \frac{R_1}{2\bar{N}},
\end{equation}
where $\bar{N}$ is the average number of transmissions per data packet.

\subsection{Numerical Results}

As a numerical example, Fig.~\ref{fig:SpectralEfficiencyDiffRate} depicts the average transmission rate $\bar{R}$ versus SNR for different rate values and a maximum number of HARQ rounds $M=3$. It can be seen that given a fixed SNR, the average transmission rate $\bar{R}$ increases as we increase $R_1$ and as the SNR increases. For instance, if we set the value of $R_1=6$ bps/Hz, the average transmission rate $\bar{R}$ increases from 1 bps/Hz at 6 dB of SNR to 2 bps/Hz at 15 dB of SNR. We recall that $\bar{R}=R_1/2\bar{N}$. At low SNR, we have $\bar{N}=M$ which implies $\bar{R}=R_1/2M$. Similarly, at high SNR, we have $\bar{N}=1$ which implies $\bar{R}=R_1/2$. Additionally, we notice that the upper bound expression for outage probability gives a tight lower bound for the average transmission rate $\bar{R}$.
\begin{figure}[h]  \centerline{\psfig{figure=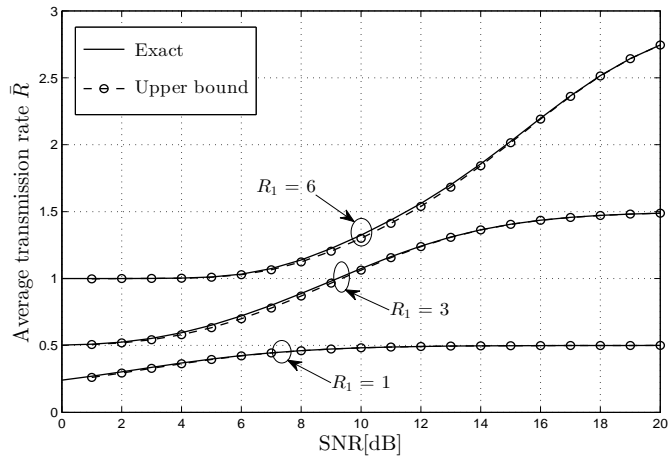,height=2.5 in,width=3.60in} }
    \caption{Average transmission rate $\bar{R}$ versus SNR for different rate values ($R_1$=~1, 3, 6~bps/Hz) and $M=3$.\vspace*{-0.4cm}}
    \label{fig:SpectralEfficiencyDiffRate}
\end{figure}

In Fig.~\ref{fig:SpectralEfficiencyVsM}, we evaluate the spectral efficiency $\bar{R}$ versus the maximum number of rounds $M$ for different SNR values and a rate $R_1=4$ bps/Hz. This figure shows that the spectral efficiency drops significantly when we increase $M$ from 2 to 4. This decay increases as the SNR decreases. For instance, when we increase $M$ from 2 to 4, the spectral efficiency loss has a value of 0.48 bps/Hz and 0.13 bps/Hz approximately for an SNR of 2 dB and 8 dB, respectively. It is worth recalling here that the outage probability decreases as $M$ increases [see Fig.~\ref{fig:OutagePrWithoutRel}]. Hence, for a constant communication rate $R_1$, as $M$ increases the spectral efficiency drops while the link reliability increases. Thus, there is a tradeoff between reliability and spectral efficiency in this case. For values of $M \geq 6$, the spectral efficiency remains constant even if we increase $M$. Finally, given a fixed $M$, we observe that spectral efficiency increases as SNR increases.

\begin{figure}[ht]  \centerline{\psfig{figure=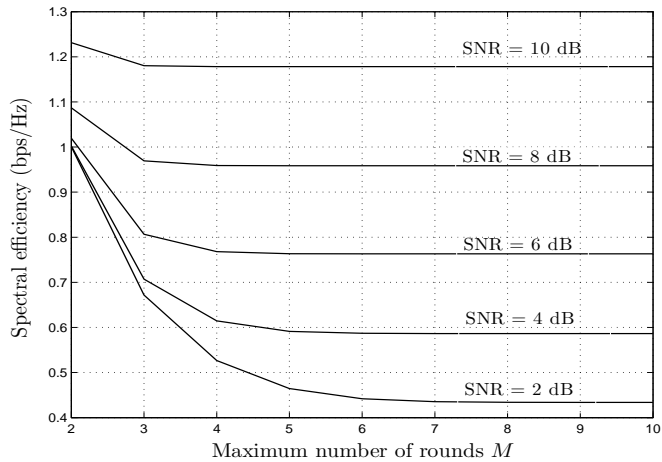,height=2.5 in,width=3.60in} }
    \caption{Spectral efficiency versus the maximum number of rounds $M$ for different SNR values and $R_1=~4$~bps/Hz.\vspace*{-5cm}}
    \label{fig:SpectralEfficiencyVsM}
\end{figure}

\section{Energy Efficiency Analysis}
\subsection{Analytical Expression}
In this section, we study the energy efficiency of HARQ-IR over amplify-and-forward dual hop relay networks. Let $\eta_{EE}$ denote the energy efficiency. It is defined as the ratio of the throughput and the average consumed power and is given as \cite{Chang2012,Li2011}
\begin{equation}
\eta_{EE} = \frac{R_1}{\bar{P}} = \frac{R_1}{\bar{N}P} = \frac{2 \bar{R}}{P},
\end{equation}
where $\bar{P}$ denotes the average consumed power, and $P$ denotes the consumed power per HARQ round.

\subsection{Numerical Results}
In Fig.~\ref{fig:EnergyEfficiencyVersusRate}, we plot the energy efficiency $\eta_{EE}$ as a function of the rate $R_1$ for different SNR values and $M=3$. This figure shows that the energy efficiency increases as the rate $R_1$ increases. Hence, for low data rate, increasing the rate will result in increasing the average number of transmissions $\bar{N}$ but with a lower slope than the rate $R_1$. For large SNR, the energy efficiency saturates as we increase the rate. For low SNR, the energy efficiency increases as the rate increases. Note that for a rate $R_1>2.5$~bps the energy efficiency increases linearly with the rate $R_1$ if we set SNR=~5 dB. We recall that the spectral efficiency increases also when the rate increases (see Fig.~\ref{fig:SpectralEfficiencyDiffRate}). So, Both energy efficiency and spectral efficiency increase when the rate increases. Additionally, we notice that energy efficiency drops when SNR increases. This fact can be seen as well from Fig.~\ref{fig:EnergyEfficiencyVersusSNR}. Finally, we observe that the upper bound expression for outage probability gives a good approximation for the exact energy efficiency.
\begin{figure}[h]  \centerline{\psfig{figure=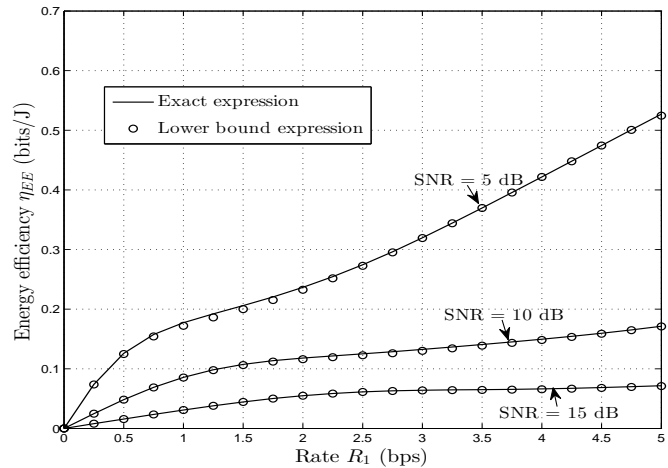,height=2.5in,width=3.60in} }
    \caption{Energy efficiency $\eta_{EE}$ versus rate $R_1$ for different SNR values (SNR=~5, 10, 15~dB) and $M=3$.}
    \label{fig:EnergyEfficiencyVersusRate}
\end{figure}

Fig.~\ref{fig:EnergyEfficiencyVersusSNR} depicts the energy efficiency $\eta_{EE}$ versus SNR for different maximum number of rounds $M$ and $R_1=1~$~bps/Hz. We notice that the energy efficiency drops when the maximum number of rounds $M$ increases especially at low SNR. Moreover, for $M \geq 5$, increasing $M$ have no impact on the energy efficiency. For a constant communication rate $R_1$, as the SNR increases the energy efficiency drops while the spectral efficiency increases. Hence, there is a tradeoff between energy efficiency and spectral efficiency.
\begin{figure}[h]  \centerline{\psfig{figure=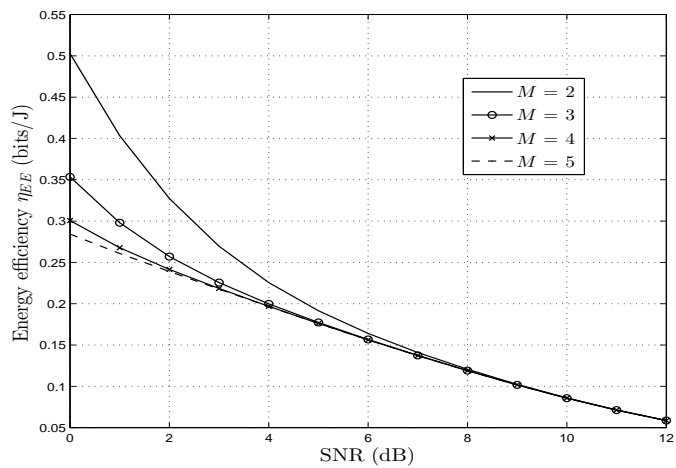,height=2.5in,width=3.6in} }
    \caption{Energy efficiency $\eta_{EE}$ versus SNR for different maximum number of rounds $M$ ($M$=~2, 3, 4, 5~rounds) and $R_1=1~$~bps/Hz.}
    \label{fig:EnergyEfficiencyVersusSNR}
\end{figure}

\section{Delay Analysis}
\subsection{Analytical Expression}
Assuming Poisson arriving packets at the source with rate $\lambda$, the average waiting time for a data packet, i.e., the time elapsed between the first sequence transmission and the successful decoding of the packet at the destination, can be evaluated using the Pollaczek-Khinchin equation \cite{Chan1997}
\begin{equation}
W = \frac{\lambda~\mathbb{E}(T_r^2)~T_F^2}{2(1-\rho)}+\frac{T_F}{2},
\end{equation}
where $T_F$ is the frame duration, and $T_r$ is a random variable accounting for the total number of transmissions for a given packet. The parameter $\rho$ should satisfy the following stability condition
\begin{equation}
\rho = \lambda~\mathbb{E}(T_r)~T_F  < 1.
\end{equation}
The packet's sojourn time in the buffer, i.e., the average waiting time in the buffer for a data packet before the start of its transmission, can be expressed as \cite{Chan1997}
\begin{equation}
T_{soj}=W+\mathbb{E}(T_r)~T_F.
\end{equation}
The expression of the average number of transmissions $\mathbb{E}(T_r)$ is provided by (\ref{Nbar}). The second-order moment $\mathbb{E}(T_r^2)$ of the total number of transmissions can be expressed as 
\begin{equation}
\mathbb{E}(T_r^2) = 1+ \sum_{m=1}^{M-1} (2m+1)P_{out}(m).
\end{equation}

\subsection{Numerical Results}
This section is devoted to the study of the average waiting time $W$ and sojourn time $T_{soj}$. The frame length is set to $T_F=1~$~s. In Figs.~\ref{fig:ExpectedWaitingTimeVersusSNR} and \ref{fig:SojournTimeVersusSNR}, we set the packet arrival rate to $\lambda=0.01$.

In Fig.~\ref{fig:ExpectedWaitingTimeVersusSNR}, we plot the expected waiting time $W(s)$ versus SNR for different rate values and $M=3$. It can be seen that the expected waiting time $W(s)$ is a concave function of the SNR. A point $P^*$ that maximizes the expected waiting time can be obtained. We notice that the point $P^*$ increases as the rate increases. Additionally, the maximum expected waiting time increases as rate increases. Thus, increasing the rate does not increase only the spectral efficiency and the energy efficiency, but also the expected waiting time. It can be seen also that the upper bound expression for outage probability gives a tight upper bound for the expected waiting time $W(s)$.
\begin{figure}[h]  \centerline{\psfig{figure=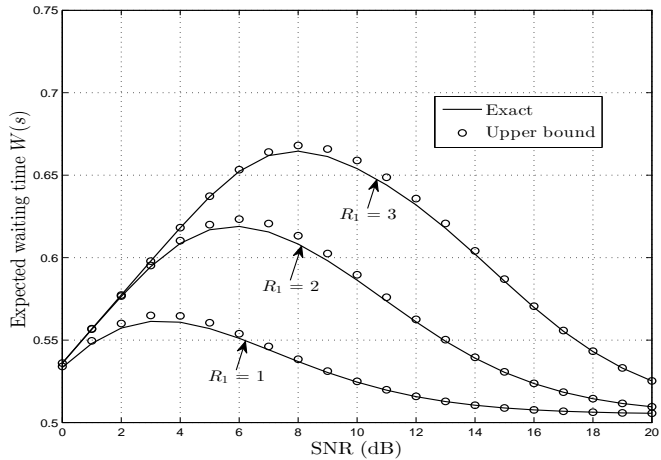,height=2.5in,width=3.60in} }
    \caption{Expected waiting time $W(s)$ versus SNR for different rate values ($R_1$=~1, 2, 3~bps/Hz) and $M=3$.}
    \label{fig:ExpectedWaitingTimeVersusSNR}
\end{figure}

Fig.~\ref{fig:ExpectedWaitingTimeVersusPacketArrivalTime} depicts the expected waiting time $W(s)$ versus packet arrival rate $\lambda$ for different maximum number of rounds $M$, a rate $R_1=1~$~bps/Hz, and SNR = 5 dB. We notice that the expected waiting time increases exponentially as the packet arrival rate increases. Additionally, we observe that the expected waiting time increases as the maximum number of rounds $M$ increases. However, this increase becomes smaller as $M$ increases. Actually, for $M \geq 5$, any increase in $M$ does not yield a significant increase in the expected waiting time.
\begin{figure}[h]  \centerline{\psfig{figure=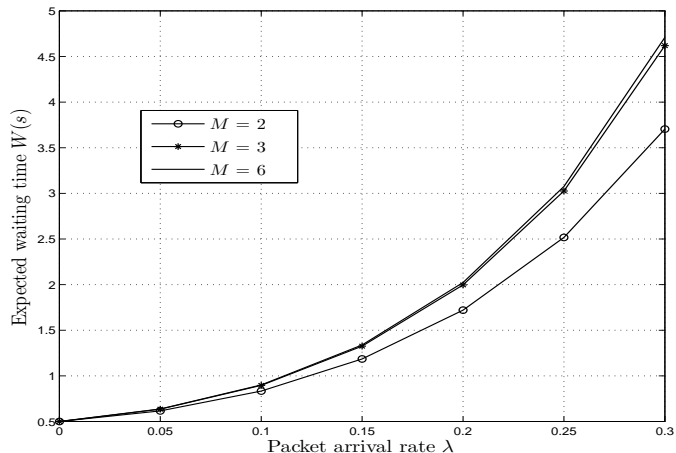,height=2.5in,width=3.60in} }
    \caption{Expected waiting time $W(s)$ versus packet arrival rate $\lambda$ for different maximum number of rounds $M$ ($M$=~2, 3, 6~rounds), $R_1=1~$~bps/Hz, and SNR = 5 dB.}
    \label{fig:ExpectedWaitingTimeVersusPacketArrivalTime}
\end{figure}

In Fig.~\ref{fig:SojournTimeVersusSNR}, we plot the sojourn time $T_{soj}(s)$ versus SNR for different maximum number of rounds $M$ and a rate $R_1=1~$~bps/Hz. It can be seen that the sojourn time decreases as SNR increases. Additionally, we notice that the sojourn time increases as the maximum number of rounds $M$ increases, especially at low SNR. However, this increase becomes smaller as $M$ increases. In fact, for $M \geq 5$, increasing $M$ have no impact on the sojourn time which conforms with the result found for the expected waiting time. We observe also that, for any value of $M$, the sojourn time $T_{soj}\to$1.5 s when the SNR$\geq$~19~dB. Finally, this figure shows also that the upper bound expression for outage probability gives a tight upper bound for the the sojourn time, and this bound becomes tighter as $M$ decreases.
\begin{figure}[h]  \centerline{\psfig{figure=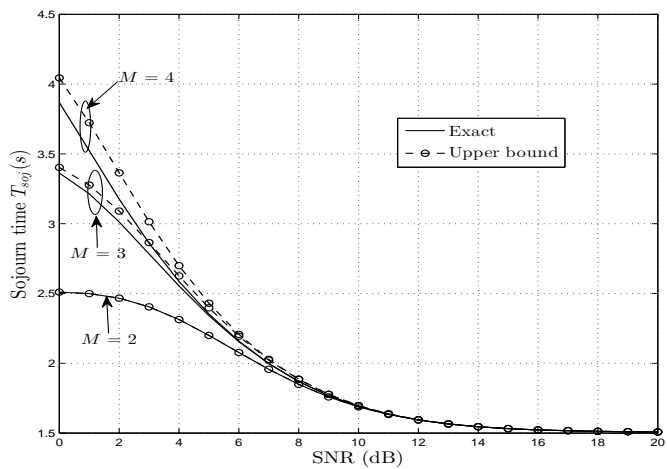,height=2.5in,width=3.60in} }
    \caption{Sojourn time $T_{soj}(s)$ versus SNR for different maximum number of rounds $M$ ($M$=~2, 3, 4~rounds) and $R_1=1~$~bps/Hz.}
    \label{fig:SojournTimeVersusSNR}
\end{figure}

\section{Conclusion}
In this paper, we have investigated the performance of HARQ-IR over dual-hop amplify-and-forward relay channels. We consider in our analysis a block fading channel. The fading varies independently from one HARQ round to another and remains constant over a single HARQ round. To limit delay a maximum number $M$ of retransmissions is assumed. We have derived an exact expression as well as an upper bound expression for the information outage probability. Our analysis reveals that the use of relaying in conjunction with HARQ yields a significant improvement in the link reliability. This latter increases as the maximum number $M$ of HARQ rounds increases. Moreover, we derive expressions for the average number of transmissions and the average transmission rate. The study of the average transmission rate $\bar{R}$ showed that if we use a constant rate $R_1$, as $M$ increases, the spectral efficiency drops while the link reliability increases. Thus, there is a tradeoff between reliability and spectral efficiency in this case.

Additionally, we have studied the energy efficiency of HARQ-IR over the amplify-and-forward dual-hop relay channel. We notice that the energy efficiency increases as the rate increases and as the SNR decreases. We carried out a delay analysis of the HARQ-IR scheme with relaying and evaluated the average waiting time and the sojourn time. We studied the impact of the rate, the packet arrival rate, and the maximum number of transmissions on the expected waiting time. We also explored the effect of the maximum number of transmissions on the sojourn time.
\appendices
\renewcommand*{\theequation}{\thesection.\arabic{equation}}
\setcounter{equation}{0}
\section{Proofs of Theorem \ref{Theo1} and Theorem \ref{Theo2}}\label{A}
In this appendix, the proofs of Theorem 1 and 2 are provided. We use the Mellin transform to derive the PDF and the CDF of the product of shifted sum of exponential random variables.
\subsection{Proof of Theorem \ref{Theo1}}\label{A.A}
Let $\xi_M=\prod_{l=1}^M Z_l=\prod_{l=1}^M (1+\gamma_l)$, where $\gamma_l$ is a random variable having as a PDF a weighted sum of exponential [see (\ref{PDFgammal})].
The PDF of $Z_l$ can be expressed as
\begin{equation}
p_{Z_l}(z)=
\begin{cases}
p_{\gamma_l}(z-1), &\text{if $z \geq 1$}\\
0, &\text{otherwise}.\\
\end{cases}
\end{equation}
Using the Mellin transform property from \cite{Springer1979}, we can express the Mellin transform of $\xi_M$ as
\begin{equation}
\mathcal{M}_s\big\{p_{\xi_M}(z)\big\}=\prod_{l=1}^M \mathcal{M}_s\big\{p_{Z_l}(z)\big\}.
\end{equation}
The Mellin transform of $Z_l$ can be derived as
\begin{align}
\mathcal{M}_s\big\{&p_{Z_l}(z)\big\} = \int_1^\infty z^{s-1} p_{Z_l}(z) dz   \nonumber\\
&=\int_1^\infty C z^{s-1} \Big( \exp\big(-\frac{z-1}{\bar{\gamma}_{up}}\big) - \exp\big(-\frac{z-1}{\bar{\gamma}_{SD}}\big)\Big) dz.
\end{align}
Using the change of variable $u=\frac{z}{\bar{\gamma}}$, we obtain
\begin{align}
\mathcal{M}_s\big\{p_{Z_l}(z)\big\} =~& C~\mathrm{e}^{\frac{1}{\bar{\gamma}_{up}}} \Big( \frac{1}{\bar{\gamma}_{up}} \Big)^{-s} \Gamma\big(s,\frac{1}{\bar{\gamma}_{up}}\big) \nonumber\\
& -C~\mathrm{e}^{\frac{1}{\bar{\gamma}_{SD}}} \Big( \frac{1}{\bar{\gamma}_{SD}} \Big)^{-s} \Gamma\big(s,\frac{1}{\bar{\gamma}_{SD}}\big),
\end{align}
where $\Gamma(\cdot,\cdot)$ is the upper incomplete gamma function defined as \cite[Eq.(06.06.02.0001.01)]{Mathematica}
\begin{equation}
\Gamma(a,z)=\int_z^\infty t^{a-1}\mathrm{e}^{-t} dt.
\end{equation}
Thus, the Mellin transform of $\xi_M$ can be determined as
\begin{align}
&\mathcal{M}_s\big\{p_{\xi_M}(z)\big\} = C^M \Big[ \mathrm{e}^{\frac{1}{\bar{\gamma}_{up}}} \Big( \frac{1}{\bar{\gamma}_{up}} \Big)^{-s} \Gamma\big(s,\frac{1}{\bar{\gamma}_{up}}\big) \nonumber\\
& \hspace{3cm} -C~\mathrm{e}^{\frac{1}{\bar{\gamma}_{SD}}} \Big( \frac{1}{\bar{\gamma}_{SD}} \Big)^{-s} \Gamma\big(s,\frac{1}{\bar{\gamma}_{SD}}\big) \Big]^M \nonumber\\
&= C^M \sum_{k=0}^M \binom{M}{k} (-1)^{M-k} \mathrm{e}^{\frac{k}{\bar{\gamma}_{up}}} \Big( \frac{1}{\bar{\gamma}_{up}} \Big)^{-ks} \Big( \Gamma\big(s,\frac{1}{\bar{\gamma}_{up}}\big) \Big)^k \times \nonumber\\
&\hspace{3cm} \mathrm{e}^{\frac{M-k}{\bar{\gamma}_{SD}}} \Big( \frac{1}{\bar{\gamma}_{SD}} \Big)^{-s(M-k)} \Big(\Gamma\big(s,\frac{1}{\bar{\gamma}_{SD}}\big) \Big)^{M-k}.
\end{align}
It follows that the PDF of $\xi_M$ can be found using the inverse Mellin transform as
\begin{align}
p_{\xi_M}(z) &= \mathcal{M}^{-1} \Big[ \mathcal{M}_s\big\{p_{\xi_M}(z)\big\} \Big] \nonumber\\
&=\frac{1}{j2\pi} \oint_\varsigma \mathcal{M}_s\big\{p_{\xi_M}(z)\big\} z^{-s} ds \nonumber\\
&=\frac{1}{j2\pi} ~C^M~  \sum_{k=0}^M \binom{M}{k} (-1)^{M-k}~ \mathrm{e}^{\frac{k}{\bar{\gamma}_{up}}}~ \mathrm{e}^{\frac{M-k}{\bar{\gamma}_{SD}}} \times \nonumber\\
&\oint_\varsigma \Big( \Gamma\big(s,\frac{1}{\bar{\gamma}_{up}}\big) \Big)^k \Big(\Gamma\big(s,\frac{1}{\bar{\gamma}_{SD}}\big) \Big)^{M-k} \Big( \frac{z}{\bar{\gamma}_{up}^k \bar{\gamma}_{SD} ^{M-k}} \Big)^{-s} ds \nonumber\\
&=C^M \sum_{k=0}^M \binom{M}{k} (-1)^{M-k} ~ \mathrm{e}^{\frac{k}{\bar{\gamma}_{up}}}~\mathrm{e}^{\frac{M-k}{\bar{\gamma}_{SD}}} \times \nonumber \\
& H_{0,M}^{M,0} \Bigg[\frac{x}{\bar{\gamma}_{up}^{k}~ \bar{\gamma}_{SD}^{(M-k)}} \Bigg\vert   \begin{array} {c} - \\ \Phi_1,\ldots,\Phi_k,\Psi_1,\ldots,\Psi_{M-k} \\ \end{array} \Bigg],
\end{align}
where the terms $\Phi_{\ell}$ and $\Psi_p$ are defined as
\begin{eqnarray}
  \Phi_{\ell} &\triangleq& (0,1,\frac{1}{\bar{\gamma}_{up}})\,\,\, \mathrm{for} \,\,\,\, \ell=1,\ldots,k\\
  \Psi_p&\triangleq& (0,1,\frac{1}{\bar{\gamma}_{SD}})\,\,\, \mathrm{for} \,\,\,\, p=1,\ldots,M-k.
\end{eqnarray}
\subsection{Proof of Theorem \ref{Theo2}}\label{A.B}
The CDF $F_{\xi_M}$ of $\xi_M$ can be deduced from its PDF as follows
\begin{align}
F_{\xi_M}(x) &=\int_0^x p_{\xi_M}(z)dz \nonumber \\
&=\frac{1}{j2\pi} C^M \sum_{k=0}^M \binom{M}{k} (-1)^{M-k} ~\mathrm{e}^{\frac{k}{\bar{\gamma}_{up}}} ~\mathrm{e}^{\frac{M-k}{\bar{\gamma}_{SD}}} \times \nonumber\\
\oint_\varsigma \Big( \Gamma\big(s,&\frac{1}{\bar{\gamma}_{up}}\big) \Big)^k \Big(\Gamma\big(s,\frac{1}{\bar{\gamma}_{SD}}\big) \Big)^{M-k} \Big[ \int_0^x \Big( \frac{z}{\bar{\gamma}_{up}^k \bar{\gamma}_{SD} ^{M-k}} \Big)^{-s} dz \Big] ds \nonumber\\
&=\frac{1}{j2\pi} ~C^M~ \sum_{k=0}^M \binom{M}{k} (-1)^{M-k}~ \mathrm{e}^{\frac{k}{\bar{\gamma}_{up}}}~ \mathrm{e}^{\frac{M-k}{\bar{\gamma}_{SD}}} \times \nonumber\\
&~~~\oint_\varsigma \Big( \Gamma\big(s,\frac{1}{\bar{\gamma}_{up}}\big) \Big)^k \Big(\Gamma\big(s,\frac{1}{\bar{\gamma}_{SD}}\big) \Big)^{M-k} \times \nonumber\\
&\hspace{2cm} \Big(\frac{\bar{\gamma}_{up}^k \bar{\gamma}_{SD} ^{M-k}}{1-s} \Big) \Big( \frac{z}{\bar{\gamma}_{up}^k \bar{\gamma}_{SD} ^{M-k}} \Big)^{-(s-1)} ds.
\end{align}
Using the relation $1-s=\Gamma(2-s,0)/\Gamma(1-s,0)$ and the change of variable $h=s-1$, we obtain
\begin{align}
\nonumber&F_{\xi_M}(x) 
=C^M \sum_{k=0}^M \binom{M}{k} \bar{\gamma}_{up}^k ~(-\bar{\gamma}_{SD})^{M-k} ~ \mathrm{e}^{-\frac{k}{\bar{\gamma}_{up}}} \mathrm{e}^{-\frac{M-k}{\bar{\gamma}_{SD}}} \times \nonumber \\
& H_{1,M+1}^{M,1} \Bigg[\frac{x}{\bar{\gamma}_{up}^{k}~ \bar{\gamma}_{SD}^{(M-k)}} \Bigg\vert   \begin{array} {c} (1,1,0) \\ \Phi'_1,\ldots,\Phi'_k,\Psi'_1,\ldots,\Psi'_{M-k},(0,1,0) \\ \end{array} \Bigg],
\end{align}
where the terms $\Phi_{\ell}$ and $\Psi_p$ are defined as
\begin{eqnarray}
  \Phi'_{\ell} &\triangleq& (1,1,\frac{1}{\bar{\gamma}_{up}})\,\,\, \mathrm{for} \,\,\,\, \ell=1,\ldots,k\\
  \Psi'_p&\triangleq& (1,1,\frac{1}{\bar{\gamma}_{SD}})\,\,\, \mathrm{for} \,\,\,\, p=1,\ldots,M-k.
\end{eqnarray}

\renewcommand*{\theequation}{\thesection.\arabic{equation}}
\setcounter{equation}{0}
\section{Proofs of Theorem \ref{Theo3} and Theorem \ref{Theo4}}\label{B}
We present in this appendix the proofs of Theorem 3 and 4.  We use the inverse Mellin transform approach to derive the PDF for the product of the sum of exponential independent random variables.
\subsection{Proof of Theorem \ref{Theo3}}\label{B.A}
Let $\zeta_M=\prod_{l=1}^M \gamma_l$. Using the Mellin transform property from \cite{Springer1979}, we can express the Mellin transform of $\zeta_M$ as
\begin{equation}
\mathcal{M}_s\big\{p_{\zeta_M}(x)\big\}=\prod_{l=1}^M \mathcal{M}_s\big\{p_{\gamma_l}(x)\big\}.
\end{equation}
The Mellin transform of $\gamma_l$ can be determined as
\begin{align}
\mathcal{M}_s\big\{&p_{\gamma_l}(x)\big\} = \int_0^\infty x^{s-1} p_{\gamma_l}(x) dx   \nonumber\\
&=\int_0^\infty C x^{s-1} \Big( \exp\big(-\frac{x}{\bar{\gamma}_{up}}\big) - \exp\big(-\frac{x}{\bar{\gamma}_{SD}}\big)\Big) dx.
\end{align}
Using the change of variable $u=\frac{x}{\bar{\gamma}}$, we obtain
\begin{align}
\mathcal{M}_s\big\{p_{\gamma_l}(x)\big\} =~& C\left((\bar{\gamma}_{up})^{s}-(\bar{\gamma}_{SD})^{s}\right) \Gamma(s),
\end{align}
where $\Gamma(s)=\int_0^{\infty} u^{s-1} \mathrm{e}^{-u} du$.

Thus, the Mellin transform of $\zeta_M$ can be expressed as
\begin{align}
&\mathcal{M}_s\big\{p_{\zeta_M}(x)\big\} = C^M\left[(\bar{\gamma}_{up})^{s}-(\bar{\gamma}_{SD})^{s}\right]^M [\Gamma(s)]^M.
\end{align}
It follows that the PDF of $\zeta_M$ can be found using the inverse Mellin transform as
\begin{align}
&p_{\zeta_M}(x)= \mathcal{M}^{-1} \Big[ \mathcal{M}_s\big\{p_{\zeta_M}(x)\big\} \Big] \nonumber\\
&=\frac{1}{j2\pi}\oint_\varsigma C^M~\big((\bar{\gamma}_{up})^{s}-(\bar{\gamma}_{SD})^{s}\big)^M~ [\Gamma(s)]^M x^{-s}ds \nonumber \\
&=\frac{C^M}{j2\pi}\sum_{k=0}^M \binom{M}{k} (-1)^{M-k} \oint_\varsigma ~[\Gamma(s)]^M ~\bar{\gamma}_{up}^{ks}~ \bar{\gamma}_{SD}^{(M-k)s}~ x^{-s}ds \nonumber \\
&=\frac{C^M}{j2\pi}\sum_{k=0}^M \binom{M}{k} (-1)^{M-k} \oint_\varsigma ~[\Gamma(s)]^M ~ \Big(\frac{x}{\bar{\gamma}_{up}^{k}~ \bar{\gamma}_{SD}^{(M-k)}}\Big)^{-s}ds \nonumber \\
&=C^M \sum_{k=0}^M \binom{M}{k} (-1)^{M-k}~ G_{0,M}^{M,0} \Bigg[~ \frac{x}{\bar{\gamma}_{up}^{k}~ \bar{\gamma}_{SD}^{(M-k)}} \Bigg\vert \begin{array} {c} - \\ 0,...,0 \\ \end{array} \Bigg],
\end{align}
where $G_{p,q}^{m,n}(\cdot)$ is the Meijer's G-function defined as (\ref{MeijerG}).
\subsection{Proof of Theorem \ref{Theo4}}\label{B.B}
The CDF $F_{\zeta_M}$ of $\zeta_M$ can be expressed as
\begin{align}
F_{\zeta_M}(x) &=\int_0^x p_{\zeta_M}(t)dt \nonumber \\
&=\frac{C^M}{j2\pi}\sum_{k=0}^M \binom{M}{k} (-1)^{M-k} \times  \nonumber \\
&\hspace{1.5cm} \oint_\varsigma ~[\Gamma(s)]^M ~ \int_0^x \Big(\frac{t}{\bar{\gamma}_{up}^{k}~ \bar{\gamma}_{SD}^{(M-k)}}\Big)^{-s}dtds \nonumber \\
&=\frac{C^M}{j2\pi}\sum_{k=0}^M \binom{M}{k} (-1)^{M-k} \times  \nonumber \\
&\hspace{1.5cm} \oint_\varsigma ~\frac{[\Gamma(s)]^M}{1-s} ~ \Big(\frac{x}{\bar{\gamma}_{up}^{k}~ \bar{\gamma}_{SD}^{(M-k)}}\Big)^{-s}ds.
\end{align}
Using the relation $1-s=\Gamma(2-s)/\Gamma(1-s)$, we obtain
\begin{align}
F_{\zeta_M}(x) &=C^M~ x~ \sum_{k=0}^M \binom{M}{k} (-1)^{M-k} \times  \nonumber \\
&\hspace{2cm} G_{1,M+1}^{M,1} \Bigg[\frac{x}{\bar{\gamma}_{up}^{k}~ \bar{\gamma}_{SD}^{(M-k)}} \Bigg\vert \begin{array} {c} 0 \\ 0,\dots,0,1 \\ \end{array} \Bigg] \nonumber \\
&=C^M \sum_{k=0}^M \binom{M}{k} (-1)^{M-k}~ \bar{\gamma}_{up}^{k}~ \bar{\gamma}_{SD}^{(M-k)} \times \nonumber \\
&\hspace{2cm} G_{1,M+1}^{M,1} \Bigg[\frac{x}{\bar{\gamma}_{up}^{k}~ \bar{\gamma}_{SD}^{(M-k)}}  \Bigg\vert \begin{array} {c} 1 \\ 1,\dots,1,0 \\ \end{array} \Bigg].
\end{align}
\balance
\bibliographystyle{IEEEtran}
\bibliography{Reference}
\end{document}